\documentclass[pra,aps,twocolumn,floatfix,showpacs,tightenlines,superscriptaddress,amsmath,amssymb]{revtex4}
\usepackage{amssymb,amsmath,amsthm,amsbsy,epsfig,color,graphicx,times,bbm}

\newcommand{\ket}[1]{\ensuremath{\left|{#1}\right\rangle}}
\newcommand{\bra}[1]{\ensuremath{\left\langle{#1}\right |}}

\newtheorem*{theorem*}{Theorem}

\begin{document}

\title{Global-to-local incompatibility, monogamy of entanglement, and ground-state dimerization: \\
       Theory and observability of quantum frustration in systems with competing interactions}

\author{S. M. Giampaolo}
\affiliation{Dipartimento di Ingegneria Industriale, Universit\`a degli Studi di Salerno,
Via Giovanni Paolo II 132, I-84084 Fisciano (SA), Italy}
\affiliation{University of Vienna, Faculty of Physics, Boltzmanngasse 5, 1090 Vienna, Austria}

\author{B. C. Hiesmayr}
\affiliation{University of Vienna, Faculty of Physics, Boltzmanngasse 5, 1090 Vienna, Austria}

\author{F. Illuminati}
\thanks{Corresponding author: fabrizio.illuminati@gmail.com}
\affiliation{Dipartimento di Ingegneria Industriale, Universit\`a degli Studi di Salerno,
Via Giovanni Paolo II 132, I-84084 Fisciano (SA), Italy}
\affiliation{CNISM Unit\`a di Salerno, I-84084 Fisciano (SA), Italy}
\affiliation{INFN, Sezione di Napoli, Gruppo collegato di Salerno, I-84084 Fisciano
(SA), Italy}

\date{September 17, 2015}

\begin{abstract}
Frustration in quantum many body systems is quantified by the degree of incompatibility between the local and global orders associated,
respectively, to the ground states of the local interaction terms and the global ground state of the total many-body Hamiltonian. This universal
measure is bounded from below by the ground-state bipartite block entanglement. For many-body Hamiltonians that are sums of two-body interaction
terms, a further inequality relates quantum frustration to the pairwise entanglement between the constituents of the local interaction terms.
This additional bound is a consequence of the limits imposed by monogamy on entanglement shareability. We investigate the behavior of local pair
frustration in quantum spin models with competing interactions on different length scales and show that valence bond solids associated to
exact ground-state dimerization correspond to a transition from generic frustration, i.e. geometric, common to classical and quantum systems alike,
to genuine quantum frustration, i.e. solely due to the non-commutativity of the different local interaction terms. We discuss how such frustration
transitions separating genuinely quantum orders from classical-like ones are detected by observable quantities such as the static structure factor
and the interferometric visibility.
\end{abstract}


\maketitle

\section{Introduction}

Many body systems are typically modeled by global Hamiltonians that are sums of local interaction terms. Frustration characterizes those situations
for which it is impossible to satisfy simultaneously the minimization of all local interaction energies in the presence of the global constraint
that imposes minimization of the total many-body energy~\cite{Toulouse1977,Kirkpatrick1977,Villain1977,Binder1986,Mezard1987,Lacroix2011,Diep2013}.

In interacting many-body systems frustration may be due either to complex geometrical configurations and/or competing interactions. These types of
frustration are {\em generic} in that they are common to both classical and quantum systems alike. On the other hand quantum systems, due to the
non-commutativity of different local Hamiltonians, may feature further, {\em genuine quantum}, sources of frustration with no classical
counterpart~\cite{Nielsen2004,Acin2007,Giampaolo2010}. In particular, quantum counterparts of classically unfrustrated systems can be
frustrated~\cite{Wolf2003,Nielsen2004,Giampaolo2010,Eisert2010,Schnack2010,Lacroix2011,Diep2013}. Therefore, for many-body quantum systems, a
universal measure of total frustration encompassing all possible sources, generic and specifically quantum, should quantify the degree of
incompatibility between global and local ground states associated, respectively, to the total many-body Hamiltonian and to the local interaction
terms.

Such a measure has recently been introduced, in Refs.~\cite{Giampaolo2011,Marzolino2013}, exactly in terms of the overlap between the ground
states of local interactions and the local reduced states obtained by partial trace from the global ground state of the total many-body
Hamiltonian. It provides a unified treatment of both generic and specifically quantum contributions to the total frustration of local interaction
terms. Quantum frustration is then singled out by the existence of a quantum lower bound to the total frustration that is realized in terms of the
ground-state entanglement between the local interacting subsystem and the remainder of the total many-body system.

In the present work we first show how frustration relates not only to the global block entanglement across a bipartition of the total many-body
system but also to the local entanglement among the constituents of the individual interaction terms. Next, we show how these relations rule
quantitatively the onset of ground-state dimerization and the formation of valence bond solids. For many-body Hamiltonians sums of two-body terms,
the local interacting pairs are naturally identified as the elementary subsystems. We deduce a direct relation between quantum frustration of local
pairs and the pair entanglement as measured by the concurrence~\cite{Hill1997,Wootters1998}. Herewith we are able to establish a relation between
quantum frustration and entanglement monogamy~\cite{Coffman2000,Osborne2006} which imposes a strong constraint, without classical counterpart, on
the shareability of quantum correlations.

Specifically, we investigate the interplay between generic and quantum frustration for an ample class of spin--$1/2$ models with nearest-neighbor
(NN) and next-to-nearest-neighbor (NNN) interactions. These models feature both quantum frustration due to the non-commutativity of the different
local interaction terms and generic frustration due to competing interactions on different length scales. We show that ground-state dimerization in
these systems is realized exactly at the transition from generic to quantum frustration, i.e. at values of the Hamiltonian parameters for which the
total frustration becomes entirely due to the ground-state block entanglement.

Remarkably, such models are amenable to quantum simulation with systems of trapped
ions~\cite{trappedions1,trappedions2,trappedions3,trappedions4,trappedions5}, ultracold atoms in optical
lattices~\cite{opticallattice1,opticallattice2,opticallattice3,opticallattice4}, coupled cavity arrays~\cite{Angelakis2007,Plenio2007},
superconducting qubits~\cite{Houk2012,Zippilli2014}, or nuclear magnetic resonance simulators~\cite{NMR1,NMR2}.
Therefore, as we will show in the following, changes in the ground-state patterns caused by changes in the frustration properties can be
experimentally observed, e.g. in optical-lattice realizations, by measuring the quasi momentum distribution as quantified by the static structure
factor~\cite{Gerbier2005,Hoffmann2009}, or the interferometric visibility~\cite{Greiner2014}.

The paper is organized as follows. In Section~\ref{Sec_frustration} we introduce the universal measure of total frustration in terms of the
global-to-local infidelity and describe its main features. In Section~\ref{Sec:entanglement} we consider generic spin--1/2 models with pairwise
interactions on different length scales, discuss their frustration properties according to the degeneracy of the local ground spaces, establish
rigorous conditions for the presence of quantum frustration, and derive an exact inequality on pairwise frustration based on entanglement monogamy.
In Section~\ref{Sec:transition} we use these results to investigate the transition from generic frustration to genuine quantum frustration in models
with competing nearest-neighbor (NN) and next-to-nearest-neighbor (NNN) interactions, and identify different phases associated, respectively,
to generic frustration and to genuine quantum frustration. In Section~\ref{Sec:observability} we discuss the observable characterization of such
frustration transitions in terms of the static structure factor and the direct observation of the universal measure of total frustration in terms
of the interferometric visibility. Summary and outlook are discussed in Section~\ref{Sec:conclusion}.

\section{Frustration}
\label{Sec_frustration}

Here we review the main aspects of a universal measure of total frustration expressed in terms of the incompatibility between local and global
orders, and briefly discuss how this universal measure (global-to-local incompatibility, or global-to-local infidelity) connects to
quantum entanglement. The class of many-body Hamiltonians
$H_{G}$ under investigation comprises all those that can be written as sums of different local terms: $H_{G}=\sum_{\ell} h_{\ell}$, where $G$ stands
for the global (total) system and $\ell \subset G$ are the local subsystems associated to the local interactions $h_{\ell}$. Let us consider one
such local subsystem $\ell$. Denoting the remainder of the total system by $R = G \setminus \ell$, frustration of $\ell$ occurs when the local
projection
of the global ground state of $H_{G}$, obtained by partial trace over $R$, i.e. the reduced state of subsystem $\ell$, fails to belong to the local
ground state space of $h_{\ell}$. Consequently, such frustration is directly quantified in terms of the overlap between the projector $\Pi_{\ell}$
onto the local ground state space of the local Hamiltonian $h_{\ell}$ and the reduced local density matrix $\rho_{\ell}=Tr_R\{ \rho_G \}$ from the
global ground state $\rho_G$ of the total Hamiltonian $H_{G}$. The ensuing universal measure of total frustration $f_{\ell}$ is thus defined as
follows~\cite{Giampaolo2011}:
\begin{equation}
\label{def_frustratio}
f_{\ell} = 1 - \mathrm{Tr} \left\lbrace(\Pi_{\ell}\otimes \mathbbm{1}_R) \, \rho_G\right\rbrace
= 1 - \mathrm{Tr} \left\lbrace\Pi_{\ell}\, \rho_{\ell} \right\rbrace \, .
\end{equation}
The measure is uniquely defined as long as $G$ features a non degenerate global ground state space, while all local subsystems $\ell$ may feature
local ground spaces of arbitrary degeneracy.

If the global ground state is degenerate, then the measure can in general depend on which global ground state is actually considered. Uniqueness is
guaranteed by considering the {\em a priori} equiprobable statistical average, i.e. the convex combination of all possible degenerate global ground
states with equal weights or, in other words, the maximally mixed global ground state (MMGGS)~\cite{Marzolino2013}.

Ground-state degeneracy in the presence of
ordered phases with nonvanishing local order parameters is tied to spontaneous breaking of some symmetries of the total Hamiltonian, while, by
definition, the MMGGS preserves all symmetries of the model. Choosing a particular global ground state among the degenerate ones
or the MMGGS depends on whether one is interested in features of the frustration that are, respectively, either state- or model-dependent. In the
following we will study frustration in different quantum spin systems featuring both non-degenerate and degenerate global ground states.

Majorization theory and the Cauchy interlacing theorem provide an immediate link between frustration and entanglement. Indeed, it is
straightforward to show that $f_{\ell}$ is bounded from below by a quantity $\epsilon^{(r)}_{\ell}$ that depends uniquely on the reduced local
density matrix $\rho_{\ell}$~\cite{Giampaolo2011,Marzolino2013}:
\begin{equation}\label{inequality}
f_{\ell} \; \geq \; \epsilon^{(r)}_{\ell} \; ,
\end{equation}
\noindent where
\begin{equation}
\label{epsilon}
\epsilon^{(r)}_{\ell} = 1 - \sum_{i=1}^r \lambda_{i}^\downarrow \; .
\end{equation}
The index $r$ denotes the rank of the local ground space projector $\Pi_{\ell}$ and $\lambda_{i}^\downarrow$ are the the first $r$ largest
eigenvalues of the reduced local density matrix $\rho_{\ell}$ in descending order.

In the case of a non-degenerate global ground state $\rho_G$ the quantity $\epsilon^{(r)}_{\ell}$ coincides with the Hilbert-Schmidt distance
between $\rho_G$ and the set of pure states with Schmidt rank less or equal to $r$~\cite{Giampaolo2011}.

This quantity is a proper entanglement monotone~\cite{Uyanik2010}. Moreover, for $r=1$, i.e. in the case of a non-degenerate ground space of the
local interaction term $h_{\ell}$, it coincides with the bipartite
geometric entanglement introduced in Ref.~\cite{Blasone2008}. This is the ground-state entanglement between the local block $\ell$ and the
remainder $R$ of the total many body system $G$ as measured by the Hilbert-Schmidt distance $D_{HS}$ of the pure ground state
$\rho_G = \ket{\Psi}_{G}\bra{\Psi}_{G}$ from the set of the pure bi-separable states $\ket{\Psi}_{bs}$ defined as
$\ket{\Psi}_{bs} \equiv \ket{\psi}_{\ell} \otimes \ket{\phi}_R$, so that
\begin{equation}
\label{geometricblockentanglement}
\epsilon^{(1)}_{\ell} = \min_{ \{ {\ket{\Psi}}_{bs}\} } D_{HS} \left( \ket{\Psi}_{G},\ket{\Psi}_{bs} \right) \, .
\end{equation}
Therefore $\epsilon^{(1)}_{\ell}$ realizes a quantum lower bound to total frustration in terms of the bipartite geometric entanglement.

The fact that the above measure of frustration is bounded from below by quantum entanglement when the many-body system features a pure,
non-degenerate global ground state is a direct consequence of the fact that for pure states all quantum correlations reduce to entanglement.
On the other hand, mixed states allow for quantum correlations more general than entanglement, such as the quantum discord, that can be present
even in separable states. One needs to keep track of this important difference when considering the general case of degenerate global ground
states. Indeed, in order to introduce a measure of frustration that should not depend on the choice of one particular state among the possible
degenerate ones, we need to introduce a reference global ground state that properly averages over all possible degenerate ground states and
preserves all symmetries of the many-body Hamiltonian.

This can be achieved either by taking a coherent quantum superposition of all the degenerate ground states with equal amplitudes, or by
introducing the maximally mixed global ground state (MMGGS), i.e. the convex combination of all the possible pure ground states with equal
probability weights. In such a mixed state, correlations are not completely characterized by quantum entanglement alone, and the quantity
$\epsilon^{(r)}_{\ell}$ turns into the sum of the ground-state block entanglement between $\ell$ and $R$, as measured by the convex roof of the
pure-state bipartite geometric entanglement, plus the classical correlations that are established between the local subsystem $\ell$ and an
ancillary system when the latter performs a generalized quantum measurement on $\ell$~\cite{Marzolino2013,Koashi2004}.

Relation~(\ref{inequality}) allows for a more refined classification of frustration in the quantum domain than merely distinguishing between
frustrated and non-frustrated systems, as in the classical domain. Namely, we can now introduce the following classification of quantum systems
according to frustration:
\begin{flushleft}
\begin{itemize}
\item[1.] Frustration-free systems (FF): 

$f_{\ell} = \epsilon^{(r)}_{\ell} = 0 \quad \forall \; \ell \; ,$

\item[2.] Inequality saturating systems (INES): 

$f_{\ell} = \epsilon^{(r)}_{\ell} > 0 \quad \forall \; \ell \; ,$

\item[3.] Inequality non-saturating systems (non-INES): 

$f_{\ell} > \epsilon^{(r)}_{\ell} \geq 0 \quad \forall \; \ell \; .$
\end{itemize}
\end{flushleft}

Accordingly, if generic frustration is absent, the system is always in an INES ground state and thus either frustration-free if the ground state
is classically correlated or genuinely quantum frustrated if it is quantum correlated. When the inequality is not saturated (non-INES) the system's
total frustration in general comprises both generic and genuine quantum contributions. Indeed, as we will show in the following, the interplay
between generic frustration and genuine quantum frustration and the transition from non-INES to INES configurations correspond to transitions
from phases with classical-like magnetic order to phases with genuine quantum order.

\section{Frustration and monogamy of entanglement}
\label{Sec:entanglement}

Let us consider translationally-invariant (periodic boundary conditions) many-body system consisting of $N$ \mbox{spin--$1/2$} elementary
constituents, featuring generic two-body (pairwise) interaction terms with the only constraint that they preserve the parity symmetry along the
three independent spin directions, so that the total many-body Hamiltonian reads
\begin{eqnarray}\label{Hamiltonian}
H=\sum_{i,j,\alpha} J_{ij}^\alpha \sigma_i^\alpha \sigma_j^\alpha = \sum_{ij} h_{ij} \; ,
\end{eqnarray}
where $\sigma_i^{\alpha}$ are the Pauli spin operators ($\alpha=x,y,z$) for the $i$-th spin and $J_{i,j}^\alpha$ is the interaction coupling
strength between spins $i,j$ along the $\alpha$ direction. For such Hamiltonians the elementary local interaction subsystems $\ell$ are naturally
identified as the ones made up by pairs of directly interacting spins with local Hamiltonian
$h_{ij}$ whose energy eigenstates are the four maximally entangled Bell states
(singlet and triplet states).

For such Hamiltonians that include, among others, the Ising, $XY$, and Heisenberg models, it is possible to introduce a quantum generalization of
the classical Toulouse criteria for frustration-free systems~\cite{Giampaolo2011,Marzolino2013}. The first step of this construction consists in
defining a prototype model. A quantum spin Hamiltonian of the type Eq.~(\ref{Hamiltonian}) is a {\em prototype} model if there exists
{\em at least} one local ground state common to all local interaction terms $h_{ij}$ and all pairwise couplings are ferromagnetic
($J_{ij}^{\alpha} \leq 0$).
Having defined the prototype model, a quantum version of the Toulouse criteria can be formulated as follows:\\
{\em i) -- The MMGGS of a prototype models is INES for all pairwise interaction terms $h_{ij}$.}\\
{\em ii) -- For all models obtained from a prototype model by local unitary operations and partial transpositions
acting on any arbitrary subsystem, the MMGGS is still INES on all pairwise interaction terms $h_{ij}$.}

The quantum Toulouse criteria allow to make predictions on the frustration of a model without requiring the
complete and exact diagonalization of the system. In fact, they allow to state powerful exact relations between the degeneracy of the local
ground spaces, entanglement, and the presence of quantum frustration in systems described by Hamiltonians of the form Eq.~(\ref{Hamiltonian}).
Indeed, a key feature of a prototype model is that all interacting pairs admit at least one local ground state in common. Then, the following holds:

\begin{theorem*}
If a many body system with pairwise local interactions satisfies the quantum Toulouse criteria and the degeneracy of the ground space common to all
local interacting pairs has dimension $r \geq 2$, then the system is necessarily frustration-free, i.e.:
$f_{\ell} = \epsilon^{(r)}_{\ell} = 0 \; \; \; \; \forall \; \ell \, .$
\end{theorem*}

\begin{proof}
If a model satisfies the quantum Toulouse criteria and the common ground space has dimension $r=2$, then all interacting pairs admit as common
ground state two out of the four maximally entangled Bell states. Obviously any linear coherent superposition of these two Bell states is again a
legitimate ground state of each local interacting pair. All possible superpositions include separable ones of the form $|i\rangle \otimes |j\rangle$.
Hence, the fully factorized state $|\Psi_{fact} \rangle = \bigotimes_{i=1}^N |i\rangle$ is a legitimate ground state of the total many-body
Hamiltonian, in close analogy to the situation that occurs at the factorization points of many-body Hamiltonians in transverse
fields~\cite{fact1,fact2,fact3}. By immediate inspection, the ground state $|\Psi_{fact} \rangle$ has vanishing frustration, single-site
entanglement, and
pair block entanglement for all interacting pairs. Moreover, since the total energy associated to $|\Psi_{fact} \rangle$ is the sum of the minimum
energies of all local pair interaction terms, any possible other total ground state must feature the same property and hence its projection onto any
pair of interacting spins must belong to the pair local ground space, implying that its frustration vanishes on all pairs. Exactly the same argument
applies also to the cases of local ground space degeneracy $r=3$ and $r=4$.
\end{proof}


We are thus left with the only non-trivial case of a local ground space degeneracy $r=1$, corresponding to local pairs of interacting spins
admitting one of the four Bell states as a local ground state. In such a case, the quantum lower bound to frustration always coincides with the
ground-state geometric entanglement between the local block (pair) $\ell=\{i,j\}$ and the rest \mbox{$R = G \setminus \ell$} of the system. Dropping
from now on the index $r=1$, it reads
\begin{equation}
\label{geoment}
\epsilon_{\ell} = 1 - \lambda_{max} \; ,
\end{equation}
where $\lambda_{max}$ is the largest eigenvalue of the reduced local density matrix $\rho_{\ell}$. From now on, without loss of generality, we will
consider the case of anti-ferromagnetically interacting pairs of spins admitting as unique, non-degenerate ground state the antisymmetric singlet
state: $|\psi^{ij}\rangle=\frac{1}{\sqrt{2}}\left( |\uparrow_i \downarrow_j \rangle - |\downarrow_i \uparrow_j  \rangle \right)$. The corresponding
projector onto the local ground space then reads
\begin{equation}
\label{localgs}
 \Pi_{\ell=ij} =|\psi^{ij}\rangle\langle \psi^{ij}| \; .
\end{equation}
Exploiting the symmetry under parity along the three spin directions, the two-body reduced local density matrix takes the general form
\begin{equation}
\label{reduceddensity}
\rho_{\ell=ij}=\frac{1}{4}  \mathbbm{1}_{ij}+ \sum_{\alpha}  g^\alpha_{ij}\; \sigma_i^\alpha \sigma_j^\alpha \; .
\end{equation}
Here $g^\alpha_{ij}$ denote the spin-spin correlation functions: $g^\alpha_{ij}=\langle S_i^\alpha S_j^\alpha\rangle$,
where $S_i^\alpha=(1/2)\; \sigma_i^\alpha$ are the spin operators and $\langle \cdot \rangle$ indicates the expectation value over the ground
state of the system. Consequently, the frustration of each interacting local spin pair reads
\begin{eqnarray}
\label{frustration}
f_{ij}& =&\frac{3}{4}+ \sum_\alpha g_{ij}^\alpha\; ,
\end{eqnarray}
and the quantum lower bound $\epsilon_{ij}$ on the total frustration $f_{ij}$ of each local interacting pair reads
\begin{eqnarray}
\label{epsilon1}
\epsilon_{ij} & = & \frac{3}{4}+\min \left\{ g_{ij}^x+g_{ij}^y+g_{ij}^z , g_{ij}^x-g_{ij}^y-g_{ij}^z \; , \right. \nonumber \\
& & \left. \quad \quad -g_{ij}^x+g_{ij}^y-g_{ij}^z \; , -g_{ij}^x-g_{ij}^y+g_{ij}^z \right\} .
\end{eqnarray}
Since the eigenvectors of $\rho_{\ell=ij}$ are the four Bell states we can establish a direct connection between the total frustration $f_{ij}$,
the quantum lower bound $\epsilon_{ij}$, i.e. the ground-state entanglement between the local interacting pair $ij$ and the remainder of the
systems, and the entanglement of formation between the spins $i$ and $j$ in the pair, as measured by the concurrence
$\mathcal{C}_{ij}$~\cite{Hill1997,Wootters1998}. We have that the following inequality holds:
\begin{equation}
\label{epsconcurrence}
 \mathcal{C}_{ij}=\max(0,1-2\epsilon_{ij}) \ge \max(0,1-2 f_{ij}) \; ,
\end{equation}
where $\mathcal{C}_{ij}$ is the concurrence of the spin pair $\{i,j\}$ and we have exploited the explicit analytic form that it takes on the set of
the four Bell states. Equality holds only when the frustration is INES, i.e. when $f_{ij}=\epsilon_{ij}$. Therefore, saturation of the inequality
implies a direct correspondence between quantum frustration, a global feature of the pair as measured by its block geometric entanglement with
respect to the remainder of the entire many-body system, and the local quantum entanglement within the pair. The more frustrated the subsystem
$\{i,j\}$ is with respect to the remainder of the many-body system, the weaker is the internal entanglement between spins $i$ and $j$ forming the
pair. Moreover, since the quantum lower bound
$\epsilon_{ij}$ is the entanglement between the local pair and the remainder of the many-body system if and only if the global ground state is pure,
from Ineq.~(\ref{epsconcurrence}) it follows for all pure global ground states that the stronger the entanglement between spins in the
local pair, the weaker is the entanglement of the pair with the remainder of the many-body system.

These relations allow for a restatement of frustration in terms of entanglement shareability. Indeed, exploiting the
Coffman-Kundu-Wootters-Osborne-Verstraete monogamy inequality~\cite{Coffman2000,Osborne2006} that sets precise constraints on the entanglement that
can be shared among many spins, we find the following {\em ``frustration monogamy''}:
\begin{eqnarray}
\label{monogamy}
\sum_{j \in L_i}\left(\max(0,1-2 f_{ij}) \right)^2 & \leq & \sum_{j \in L_i} \mathcal{C}_{i,j}^2 \nonumber \\
& \leq &  \sum_{j \in G \setminus i} \mathcal{C}_{i,j}^2 \leq \tau_i \; .
\end{eqnarray}
In the above, $L_i$ denotes the subset of spins that interact with spin $i$, while $G \setminus i$ denotes the total spin system without spin $i$.
The upper bound $\tau_i$ is the tangle~\cite{Coffman2000,Osborne2006}. When the global ground state is non degenerate,
it coincides with the squared concurrence measuring the single-site entanglement of spin $i$ with
all other spins in the system. For spin systems described by Hamiltonians of the form Eq.~(\ref{Hamiltonian})
the tangle takes the expression
$\tau_i = 1 - m_x^2 - m_y^2 - m_z^2$, where $m_\alpha = \langle \sigma_i^\alpha \rangle$ is the local magnetization along the $\alpha$ direction.
Therefore $\tau_i = 1$ in the MMGGS, since the latter preserves all symmetries of the total many body Hamiltonian and therefore does not allow for a
spontaneous magnetization along any spin direction. This fact and the frustration monogamy
Ineq.~(\ref{monogamy}) imply exact bounds on the {\em minimum} possible amount of frustration in a many-body system. Considering for simplicity the
uniform case $f_{ij} = f \; \; \forall \{i,j\}$, and denoting by $N_{L}$ the number of spins interacting with spin $i$, we have:
\begin{equation}
\label{minimumfrustration}
f \geq \frac{1}{2} \left( 1 - \frac{1}{\sqrt{N_{L}}} \right) \; .
\end{equation}
We thus find that the minimum frustration threshold varies from $f_{min} = 0$ for $N_{L} = 1$ for a many-body system consisting in just a single
pair of interacting spins, to $f_{min} = 1/2$ at the thermodynamic limit of fully connected many-body systems, i.e. with $N_L \rightarrow \infty$.

\section{Frustration and dimerization}

\label{Sec:transition}

The frustration $f_{ij}$ of the local pair interaction $h_{ij}$ is in general a function of the parameters entering in the total many-body
Hamiltonian. Moreover, the interplay between different sources of frustration may lead to significant changes in the structure of the global
many-body ground state. In particular, we wish to investigate the changes in the structure of quantum ground states when frustration turns from
non-INES to INES, i.e. when genuine quantum frustration overrules completely other generic sources of frustration. We will focus on the
paradigmatic class of frustrated one-dimensional models with competing nearest-neighbor (NN) and next-to-nearest-neighbor (NNN)
anti-ferromagnetic interactions, the so-called $J_1 - J_2$ models~\cite{TransitionJ1J2}:
\begin{eqnarray}
\label{XXZnew}
\!\!H \!\! &\!=\!&\! \!\cos\!\phi \!\sum_{i} \cos\!\delta\!\left( S_i^x S_{i+1}^x\! +\! S_i^y
S_{i+1}^y \right)\! + \sin\!\delta S_i^z \!S_{i+1}^z \nonumber \\
\!\!\!\!&\!+\! &\! \!\sin\!\phi\; \!\!\sum_{i} \cos\!\delta\!\left( S_i^x S_{i+2}^x\! +\! S_i^y
S_{i+2}^y \right)\! + \sin\!\delta  S_{i}^z\! S_{i+2}^z  \; .
\end{eqnarray}
These systems are characterized by the ratio between NNN and NN interaction strengths $J_2/J_1 = \tan \phi$, with $\phi\in[0,\frac{\pi}{2}]$),
and the ratio between  anisotropies $\tan \delta$, with $\delta \in[0,\frac{\pi}{2}]$.
The translationally invariant case (periodic boundary conditions) is studied by exact diagonalization with $N=24$ spins.
In Fig.~\ref{ineszone} we report the behavior of the pair frustration $f_{ij}$, according to the classification induced by
Ineq.~(\ref{inequality}), as a function of the spin-spin correlation functions, where we have exploited
Eqs.~(\ref{reduceddensity}), (\ref{frustration}), and (\ref{epsilon1}), and the symmetry $g_{ij}^x = g_{ij}^y$. Both non-INES
(generic) and INES (genuine quantum) frustrations are realized in a significant range in the parameter space.

\begin{figure}
\includegraphics[width=6.5cm]{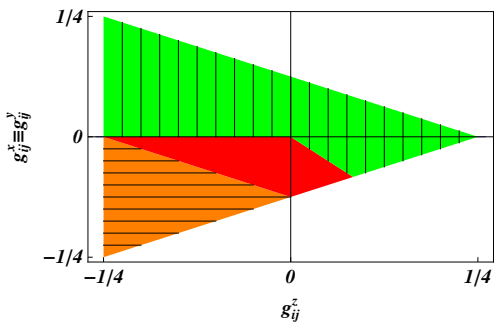}
\caption{(Color online) Behavior of Ineq.~(\ref{inequality}) for the local pair frustration as a function of the admissible values of the
correlation functions $g_{ij}^x = g_{ij}^y$ and $g_{ij}^{z}$ in the reduced local pair density matrix $\rho_{\ell=ij}$ of an $XXZ$ model with
competing NN and NNN interactions. Vertically striped green area: reduced states $\rho_{\ell=ij}$ with non-INES frustration (strict inequality,
generic frustration). Horizontally striped orange area: reduced states with INES frustration (inequality saturated, genuine quantum frustration).
Red area: disentangled reduced states with INES frustration (inequality saturated, vanishing pair concurrence, enhanced genuine quantum frustration).
Remaining white area: non-physical states.}
\label{ineszone}
\end{figure}

In Fig.~\ref{MG} we report schematically the structure of the quantum ground states for models with competing NN ($J_1$) and NNN ($J_2$)
interactions, in the class of Eq.~(\ref{XXZnew}). Exact ground-state dimerization is realized at the Majumdar-Ghosh point $\phi_{MG}
\simeq 0.45$ solution of $\tan\phi = 1/2$~\cite{MajumdarGhosh}. This is the only integrable point of models in the class Eq.~(\ref{XXZnew}).
At this point the ground state becomes doubly degenerate, breaking explicitly the translational invariance of the lattice, which is recovered at
double lattice spacing $2a$. The two ground states are tensor products of NN spin singlets, realizing the first and simplest paradigmatic instance
of valence bond states~\cite{Anderson,Dimer,AKLT}. Except for the Majumdar-Ghosh point, the global ground state is always unique. At $\phi_{MG}$
frustration is evaluated in each of the two dimerized ground states as well as in the MMGGS, i.e. their convex combination with equal weights.

\begin{figure}
\includegraphics[width=6.5cm]{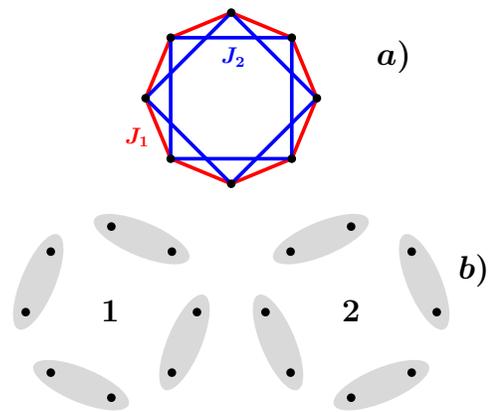}
\caption{(Color online) a) Schematic structure of a one-dimensional quantum spin model with competing NN ($J_1$) and NNN ($J_2$) anti-ferromagnetic
interactions (periodic boundary conditions). b) Ground-state structure at the integrable Majumdar-Ghosh point $J_2/J_1 = 1/2$. At this point the
model features two exactly dimerized ground states that are the tensor product of NN spin singlets. The two valence bond solids break
the translational invariance of the Hamiltonian with respect to the lattice spacing $a$.}
\label{MG}
\end{figure}

In Figs.~\ref{figfrust_r_1} and \ref{figfrust_r_2} we report the behavior of the pair frustration $f_{ij}$, the quantum lower bound
$\epsilon_{ij}$, and the pair concurrence $C_{ij}$ as functions of $\phi$, respectively for the NN ($i,i+1$) and the NNN ($i,i+2$) local pair
interaction terms. The plots are taken at anisotropies $\delta=0$ ($XX$ symmetry), $\frac{\pi}{6}$ ($XXZ$ symmetry), and $\delta=\frac{\pi}{2}$
(classical Ising symmetry). The fully isotropic Heisenberg case ($XXX$) realized by $\delta = \frac{\pi}{4}$ is reported in
Fig.~\ref{figfrust_Heisenberg} for which we provide also insets comparing the detailed behavior of the pair frustrations in the MMGGS and in the
two degenerate dimerized pure ground states at the Majumdar-Ghosh point $\phi_{MG}$.

Remarkably, from Figs.~\ref{figfrust_r_1} and \ref{figfrust_r_2}, we see that for all quantum models ($XX$, $XXZ$, and Heisenberg)
it is always $f_{i,i+1} = \epsilon_{i,i+1}$, i.e. frustration of the NN interacting local pair is always INES (genuine quantum), notwithstanding
the presence of a source of frustration, due to geometry, in terms of competing interactions on different length scales.

The NN pair frustration increases for increasing relative strength of the NNN
interactions, reaching asymptotically the maximum value $3/4$. This coincides with the maximum possible value of $\epsilon_{i,i+1}$, i.e.
the maximum possible value of the ground-state block entanglement between the NN pair and the rest of the system. Such a maximum is achieved for a
reduced pair density matrix $\rho_{i,i+1}$ with all four eigenvalues degenerate and equal to $1/4$.

In other words, maximizing the relative strength of the NNN interaction maximizes the ground-state block entanglement $\epsilon_{i,i+1}$
between NN pairs and the rest of the system and the corresponding (genuine quantum) NN pair frustration $f_{i,i+1}$. Hence, by monogamy, it
minimizes the pairwise entanglement $C_{i,i+1}$ between the spins forming the NN interacting local pair. Indeed, for a sufficiently large value
of the ratio $J_2/J_1$,
the NN pair concurrence $C_{i,i+1}$ vanishes exactly and remains zero thereafter.

\begin{figure}
\includegraphics[width=8.cm]{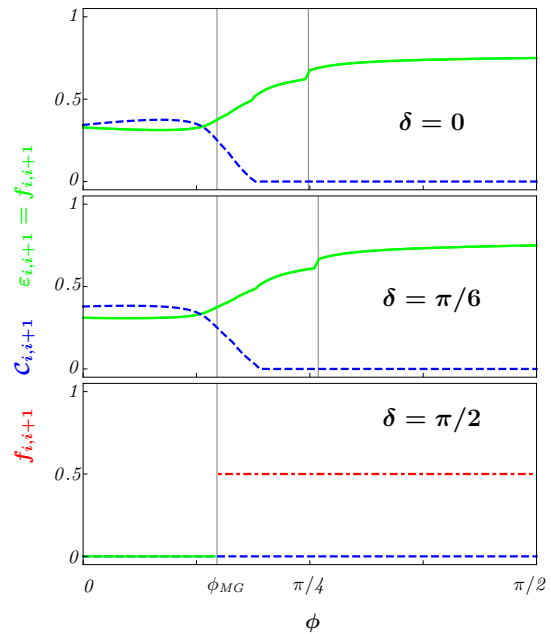}
\caption{(Color online) Frustration $f_{i,i+1}$ of the NN interacting local pair, quantum lower bound
$\epsilon_{i,i+1}$, and NN pair concurrence $C_{i,i+1}$, as functions of the parameter
$\phi$ ruling the ratio $\tan \phi$ between NN and NNN interactions. Upper panel: $J_1 - J_2$ model with $XX$ symmetry, $\delta = 0$. Solid green
line: $f_{i,i+1} = \epsilon_{i,i+1}$, as the two quantities in this case coincide $\forall \phi$. Dashed blue line: $C_{i,i+1}$. Central panel:
$J_1 - J_2$ model with $XXZ$ symmetry, $\delta = \pi/6$. Also in this case $f_{i,i+1} = \epsilon_{i,i+1} \, \forall \phi$. Both are denoted by the
solid green line. Dashed blue line: $C_{i,i+1}$. Lower panel: $J_1 - J_2$ model with Ising symmetry, $\delta= \pi/2$. Dot-dashed red line:
$f_{i,i+1}$ between the Majumdar-Ghosh point $\phi_{MG}$ and $\phi = \pi/2$. Solid green line: $f_{i,i+1} = \epsilon_{i,i+1} = 0$ between the
Majumdar-Ghosh point $\phi_{MG}$ and $\phi = 0$. Dashed blue line: $C_{i,i+1}$. NN frustration is always INES, $f_{i,i+1} = \epsilon_{i,i+1}$,
for the quantum $XX$ and $XXZ$ models. In the regime of dominating NNN interactions, $f_{i,i+1} \rightarrow 3/4$, which is the maximum possible
block entanglement $\epsilon_{i,i+1}$ between the NN interacting local pair and the rest of the system. The leftmost vertical line corresponds to
the Majumdar-Ghosh point $\phi_{MG}$. The rightmost vertical lines correspond to the boundary value $\phi_b$ of the intermediate quantum regime:
$\phi_b \simeq 0.99$ and $\phi_b \simeq 1.08$, respectively for the $XX$ and the $XXZ$ symmetries. Models with Ising symmetry feature only a
frustration due to geometry. It undergoes a sharp transition from $f_{i,i+1} = 0$ to $f_{i,i+1} = 1/2$ exactly at $\phi_{MG}$ and remains constant
thereafter. The quantities $\epsilon_{i,i+1}$ and $C_{i,i+1}$ are always vanishing in the Ising case.}
\label{figfrust_r_1}
\end{figure}

Even more important, Figs.~\ref{figfrust_r_1} and \ref{figfrust_r_2} show that frustration $f_{i,i+2}$ of the NNN interacting local pair undergoes
a full transition from non-INES (generic frustration) to INES (genuine quantum frustration) exactly at the Majumdar-Ghosh point $\tan \phi =1/2$.
Therefore, transition to exact ground-state dimerization corresponds to exact transition from non-INES to perfect INES frustration of the NNN
local interaction terms. For increasing relative strength of the NNN interactions the NNN pair frustration decreases, reaching asymptotically its
minimum (whose actual value depends on $\delta$).
Such transitions in frustration exhibit universal features in the $XX$, $XXZ$, and Heisenberg cases, irrespective of the corresponding
different symmetry classes.

As shown in Figs.~\ref{figfrust_r_1} and \ref{figfrust_r_2}, the frustration pattern differs radically only in the Ising case, as should be
expected, since all local interactions commute in Ising systems, entanglement vanishes, and, as discussed in Section~\ref{Sec:entanglement},
there are no genuine quantum sources of frustration. In this latter case we have simply two well distinguished regimes separated by the
Majumdar-Ghosh point $\phi_{MG}$. Indeed, at this point the NN
frustration jumps from zero to half the maximum value $1/2$ and remains constant thereafter, while the NNN frustration jumps from the
maximum value $1$ to a vanishing value and remains zero thereafter.

\begin{figure}
\includegraphics[width=8.cm]{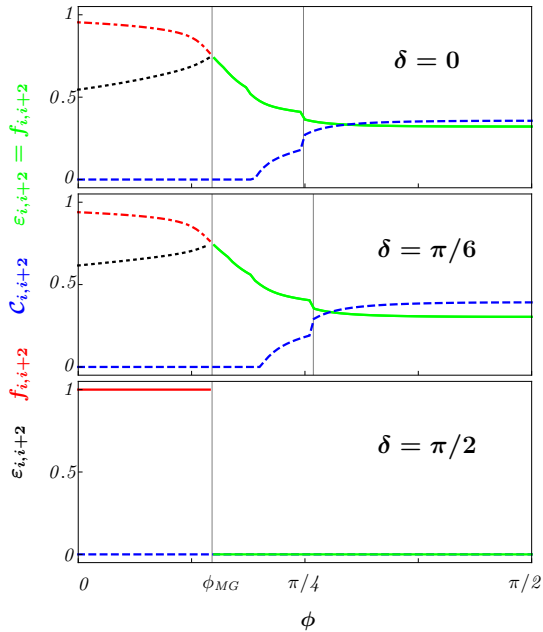}
\caption{(Color online) Same as in Fig.~\ref{figfrust_r_1} but this time for the frustration $f_{i,i+2}$, quantum lower bound $\epsilon_{i,i+2}$,
and concurrence $C_{i,i+2}$ of the NNN interacting local pairs. Transition to genuine quantum frustration (non-INES to INES),
$f_{i,i+2} = \epsilon_{i,i+2}$, occurs exactly at the Majumdar-Ghosh point. Frustration remains INES thereafter. In the regime of dominating NNN
interactions, $f_{i,i+2} \rightarrow 1/4$, which is the minimum possible block entanglement $\epsilon_{i,i+2}$ between the NNN interacting local
pair and the rest of the system. Models with Ising symmetry feature only frustration due to geometry. It undergoes a sharp transition from
$f_{i,i+2} = 1$ to $f_{i,i+2} = 0$ exactly at the Majumdar-Ghosh point and remains vanishing thereafter.}
\label{figfrust_r_2}
\end{figure}

On the other hand, in the nontrivial quantum cases with non-commuting local interactions ($XX$, $XXZ$, Heisenberg) we observe a rich structure
with three different patterns of frustration as the ratio of NNN to NN interactions varies. At small $\phi$ well below the Majumdar-Ghosh point
$\phi_{MG}$ (dominating NN interactions) NN spin pairs are weakly frustrated in opposition to NNN pairs that tend to be strongly frustrated.
Indeed, for $\phi \rightarrow 0$ the NNN pair frustration achieves values very close to unity, i.e. the maximum possible value, which is achieved
exactly in the classical Ising case. For $\phi$'s far above the Majumdar-Ghosh point $\phi_{MG}$ (dominating NNN interactions) the situation
reverses: NNN pairs become weakly frustrated while NN pairs become strongly frustrated.

\begin{figure}
\includegraphics[width=8.cm]{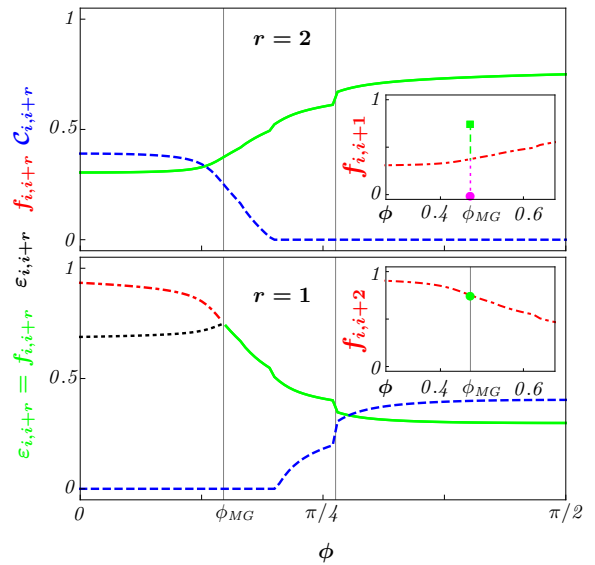}
\caption{(Color online) Same as in Figs.~\ref{figfrust_r_1} and Fig.~\ref{figfrust_r_2} but this time for $J_1 - J_2$ models with Heisenberg
($XXX$) symmetry, $\delta = \pi/4$. In this case the boundary value of the intermediate quantum region is $\phi_b \simeq 1.18$. Insets:
zoom of the behavior of the NN and NNN pair frustrations at and around the Majumdar-Ghosh point. Away from $\phi_{MG}$ the pure global ground
state is unique and thus always coincides with the MMGGS. Exactly at $\phi_{MG}$, besides the averaged ones in the MMGGS, we also report the
non-averaged NN and NNN pair frustrations in the two exactly dimerized pure global ground states. The NN pair frustration
$f_{i,i+1}$ vanishes in the ground state for which $\{i,i+1\}$ belong to the same dimer and are thus maximally entangled between themselves
(purple full circle). Vice versa, it coincides with the block entanglement $\epsilon_{i,i+1}$ at its maximum value $3/4$ in the ground state for
which
$\{i,i+1\}$ belong to different dimers (green full square). The NNN pair frustration $f_{i,i+2} = 3/4$ (green full circle) in both ground states,
since $\{i,i+2\}$ always belong to different dimers.}
\label{figfrust_Heisenberg}
\end{figure}

These two limiting regimes are characterized by two simple antiferromagnetic orders that essentially coincide with those of the classical Ising case. Fig.~\ref{lattices} summarizes this behavior pictorially in terms of the lattice structure. In the limit $\phi \rightarrow 0$ the lattice reduces to a single chain with lattice spacing $a$ and NN anti-ferromagnetic interactions, allowing for a simple anti-ferromagnetic order. In the limit $\phi \rightarrow \pi/2$ the lattice effectively splits into two chains, each with lattice spacing $2a$ and NN anti-ferromagnetic interactions. Two simple anti-ferromagnetic orders are then established on each chain.

However, in the quantum models a further intermediate regime is established in the central interval starting at the Majumdar-Ghosh point and extending to the right, i.e. with increasing values of $\phi$. In this interval, we observe a smooth crossover in terms of a succession of small but sizable step-jumps in the behavior of the NN and NNN pair frustrations. This intermediate, central regime corresponds to the onset of a quantum order that is well distinguished from the two classical (Ising-like) anti-ferromagnetic regimes discussed above. The lower and upper values of this interval are, respectively, the Majumdar-Ghosh point and the boundary point $\phi_b$ corresponding the last step-jump in the increasing behavior of $f_{i,i+1}$ and decreasing behavior of $f_{i,i+2}$.

The fact that a sharp change in the behavior in the frustration properties appears exactly at the Majumdar-Ghosh point may seem at first
surprising, since the transition to the dimerized phase in the $J_1 - J_2$ model, regardless of the anisotropy, occurs at a critical value
$\phi=\phi_c$ much lower than $\phi_{MG}$. For example, in the case of the fully isotropic Heisenberg model, $\delta=\pi/4$, the critical value
$\phi_c \simeq \tan^{-1}(0.2411)\simeq 0.237$, as determined by extensive numerical investigations based on the density matrix renormalization
group (DMRG) and other numerical algorithms~\cite{TransitionJ1J2}. However, the frustration does not feature any sharp change when crossing
$\phi_c$.

The smooth behavior of the frustration at $\phi_c$ and the appearance of step-jumps in the frustration when the system is in the intermediate quantum regime deserve further comments. Let us first consider the physical origin of the step-jumps. One might conjecture that they are due to finite-size effects and that they should disappear when considering much larger chains. In fact, this is not the case, as can be already understood by considering chains of relatively small size. Indeed, by looking at the structure of the ground state in correspondence to the occurrence of the step-jumps, we have verified by analyzing the numerical data obtained from exact diagonalization that at these points exact crossovers are realized between ground states of different parity. These crossovers are then responsible for the sudden increasing of the NN frustration and the corresponding sudden decreasing of the NNN frustration that are observed in the interval $[\phi_{MG}, \phi_b]$.

This structure of ground-state crossovers appears similar to the one occurring in $XY$ chains in transverse field $h$, for values of $h$ below the factorizing field $h_f$~\cite{fact1,fact2,fact3}. Indeed, one can draw a very close parallel between the two-site factorization of the ground state of the $J_1 - J_2$ model at the Majumdar-Ghosh exact dimerization point $\phi_{MG}$, and the single-site factorization of the ground state of the $XY$ model at the exact factorizing field $h_f$. Firstly, in both cases, the sharp change in the behavior of various physical quantities, including frustration and ground-state entanglement, that occurs exactly at the respective factorization points (single-site factorization in the $XY$ model at $h_f$, two-site factorization in the $J_1 - J_2$ model at $\phi_{MG}$), is always present, independently of the size of the system~\cite{fact3,finitesize1,finitesize2}. Moreover, in the finite-size $XY$ chain, the crossovers occur between ground states of different parity and appear for $0 < h < h_f$. Such crossovers are always present, but their number and position depend on the size of the system~\cite{fact3}. The crossovers terminate exactly at $h = 0$ and $h = h_f$. At each crossover, several physical quantities, including frustration and ground-state entanglement, feature step-jumps completely analogous to the ones that are featured in the central quantum regime of the finite-size $J_1 - J_2$ model, i.e. in the region $\phi_{MG} < \phi < \phi_b$. The analogy is complete by observing that the crossovers terminate exactly at $\phi = \phi_{MG}$ and $\phi = \phi_b$, and that both exact ground state factorization and exact ground state dimerization are not associated to any quantum phase transition. Indeed, as already mentioned, both the Majumdar-Ghosh point in the $J_1 - J_2$ model and the factorization point in the $XY$ model occur independently of the size of the system. Moreover, both points are always associated to a double degeneracy of the global ground state.

In the $XY$ model, where all quantities can be evaluated analytically and controlled precisely both at finite-size and in the thermodynamic limit, it is possible to verify exactly that the number of the step-jumps increases while the height of the jumps decreases as the size of the chain increases; in the thermodynamic limit the height of the jumps vanishes exactly and a fully analytic behavior is restored. Unfortunately, the non-solvability of the $J_1 - J_2$ model prevents us, at least at the moment, to establish whether the central quantum region extends indefinitely, i.e. the boundary point $\phi_b$ increases as the size of the chain diverges, and whether analyticity is restored in the thermodynamic limit.

Concerning the smooth behavior of the frustration observed at $\phi_c$ in small chains, we remark that the dramatic change to {\em exact} ground state dimerization occurs at and only at $\phi_{MG}$. Furthermore, our investigation relies on exact diagonalization, which is feasible only for chains of rather limited size. To address larger sizes, one would need to resort to DMRG and related algorithms such as matrix product states (MPS). However, the existing codes are not suitable, as they stand, for the evaluation of the frustration measure. The first stumbling block is that the DMRG algorithm needs profound modifications and improvements in order to compute reduced density matrices and projectors onto the local ground spaces, their eigenvectors, and their spectra, in the presence of periodic boundary conditions. The second stumbling block is that, while DMRG allows to evaluate without much effort the correlation functions between neighboring sites, it works much less efficiently when one needs to compute correlation functions between non-neighboring spins in the chain. Unfortunately, exactly such correlation functions are those needed in order to evaluate, for instance, the static structure factor. These difficulties
have motivated us, whenever direct exact analytical evaluation proved to be impossible, to resort to exact diagonalization algorithms based on
augmented Lanczos methods.

In case of sufficiently large chains, the frustration measure is indeed expected to detect the transition to approximate dimerization by showing an inflection at $\phi_c$, in the same way as it does in exactly solvable models that we are currently verifying, such as the 1-D $XY$ model, or the 1-D Kitaev and cluster-Ising models of symmetry-protected topological order~\cite{topfrust}. However, in the case of the $J_1-J_2$ model this behavior can become evident only once more powerful numerical techniques are developed in order to evaluate the frustration in chains of sizes much larger than the ones that are allowed by resorting to exact diagonalization.

\begin{figure}
\includegraphics[width=6.5cm]{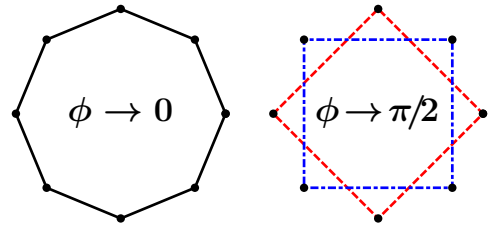}
\caption{(Color online) Lattices effectively realized in the different limits of $\phi$. For $\phi \rightarrow 0$ NNN interactions are strongly
suppressed and
the lattice reduces to a NN interacting chain with lattice spacing $a$. For $\phi \rightarrow \pi/2$ the lattice splits into two independent NN
interacting chains, each with lattice spacing $2a$.}
\label{lattices}
\end{figure}

\section{Observability: structure factor and interferometric visibility}

\label{Sec:observability}

The nature of the intermediate, central {\em quantum regime} in comparison to the two standard anti-ferromagnetic orders is best captured by studying the
behavior of
the static structure factor, as defined below and reported in Fig.~\ref{ssf_figure}. Indeed, as changes in the patterns of quantum frustration
imply a transition between different lattice geometries, we can connect them to the quantum momentum distribution observed in experiments,
which is quantified by the
static structure factor $S_f(k)$:
\begin{equation}
\label{ssf}
S_f(k)=\frac{2}{N} \sum_{i,j=1}^{N} \cos (k a |i-j| ) \langle \overrightarrow{S_i} \cdot  \overrightarrow{S_j}\rangle \; ,
\end{equation}
where $a$ is the lattice spacing, $k$ is the wave vector, and $\overrightarrow{S_i}$ is the spin operator vector of the $i$th whose components
are the three Pauli spin-1/2 matrices. In most experimental situations direct access to the correlation functions for each pair can be extremely
challenging. In such cases an ``averaged'' and ``collective'' information such as the static structure factor, that can be obtained from the
analysis of the visibility which quantifies the contrast of the time-of-flight images~\cite{Gerbier2005,Hoffmann2009}, may provide useful
coarse-grained information on the changes in the quantum orders due to changes in the frustration patterns.

In Fig.~\ref{ssf_figure} we report the behavior of the static structure factor for three different values of $\phi$, respectively well below,
around, and well above the Majumdar-Ghosh point $\phi_{MG} \simeq 0.45$, comparing the quantum $XX$ and Heisenberg models with the classical
Ising case. For $\phi \simeq 0.14$, well below $\phi_{MG}$ quantum and Ising models all features a single peak at $k=\pi/a$, corresponding to the
same anti-ferromagnetic order. At $\phi \simeq 1.4$, i.e. well above $\phi_{MG}$ all models feature two symmetric peaks at $k=\pi/2a$ and
$k=3\pi/2a$, corresponding to two anti-ferromagnetic orders, one for each of the two effective chains of Fig.~\ref{lattices}. These are the
geometric regimes for which classical Ising and quantum $XX$ and Heisenberg orders are indistinguishable.

\begin{figure}
\includegraphics[width=8.cm]{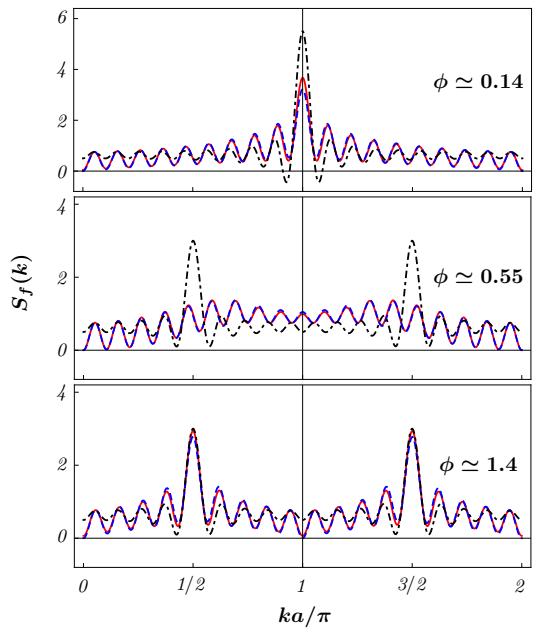}
\caption{(Color online) Behavior of the quasi-momentum distribution $S_f(k)$ as a function of the quasi-momentum $k$ for different values of $\phi$.
Black dot-dashed curve: classical $J_1 - J_2$ model with Ising symmetry ($\delta \pi/2$). Red solid curve: quantum $J_1 - J_2$ model with $XX$
symmetry ($\delta = 0$). Blue dashed curve: quantum $J_1 - J_2$ model with Heisenberg symmetry ($\delta = \pi/4$). Upper panel: dominating NN
interactions ($\phi \simeq 0.14$). The static structure factor features a single peak for all models, corresponding to a simple Ising-like
anti-ferromagnetic order in a simple chain with NN anti-ferromagnetic interactions and lattice spacing $a$. Lower panel: dominating NNN interactions
($\phi \simeq 1.4$). The lattice splits in two sub-chains, each with lattice spacing $2a$. Correspondingly, two symmetric peaks are featured by all
models, corresponding to simple anti-ferromagnetic order on each sub-chain. Central panel: intermediate regime ($\phi \simeq 0.55$). In this regime,
the classical model still features two simple anti-ferromagnetic orders, while the quantum models feature a nontrivial quantum
order and no peaks (the many small secondary peaks are due to finite-size effects that vanish in the thermodynamic limit).}
\label{ssf_figure}
\end{figure}

Finally, for $\phi \simeq 0.55$, a value above but comparable to $\phi_{MG}$, we observe that the quantum models do not feature any peak,
corresponding to a quantum order totally distinguished from the two classical anti-ferromagnetic orders featured by the Ising model on the two
effective chains.

On the other hand, although the analysis of the static structure factor provides useful information about the behavior of the frustration measure,
in general it cannot yield complete information about the relation between the frustration and its quantum lower bound $\epsilon_{ij}$. Indeed,
while the existence of a direct relationship between the static structure factor and the frustration is due to a qualitatively similar dependence of
both quantities on sums of correlation functions, a fact that can be immediately appreciated seen comparing Eq.~(\ref{frustration}) with
Eq.~(\ref{ssf}), a direct connection between correlation functions and $\epsilon_{ij}$ is hindered by the highly nonlinear nature of their mutual
relationship, where the nonlinearity originates from the nontrivial minimization that occurs in Eq.~(\ref{epsilon1}).

As already mentioned, the genuine quantum regime established at the Majumdar-Ghosh point $\phi_{MG}$ extends up to a critical boundary
value $\phi_{b}$ corresponding to the last step-jumps in the pattern of increasing NN pair frustration and decreasing NNN pair frustration, as
reported in Figs.~\ref{figfrust_r_1}, \ref{figfrust_r_2}, and \ref{figfrust_Heisenberg}. The extension of the quantum region is a non-universal
feature, as it increases, as should be intuitively expected, with increasing symmetry of the different quantum models. Indeed, for $\delta = 0$
($XX$), $\phi_b \simeq 0.99$. For $\delta = \pi/6$ ($XXZ$), $\phi_b \simeq 1.08$. Finally, for $\delta = \pi/4$
(Heisenberg), $\phi_b \simeq 1.18$.

Concerning observability, besides measuring the static structure factor in time-of-flight experiments, a complementary
strategy would be the measurement of the visibility in interferometric experiments with atoms~\cite{Greiner2014}. In quantum optics, the
interferometric visibility $V$ by feeding two successive photons from a source into a 50:50 beam splitter is
\begin{equation}
\label{visibility}
V = \mathrm{Tr} \left( \rho_a \rho_b \right) \; ,
\end{equation}
where $\rho_a$ and $\rho_b$ are the states of the two incoming photons.

In fact, it is possible to envisage very general schemes for the estimation and the direct measurement of linear and nonlinear functionals of
quantum states, based on simple quantum networks~\cite{Ekert2002}. Such schemes have been recently investigated and specialized to quantum
simulators realized with atomic ensembles in optical lattices, either in terms of quantum switches~\cite{Abanin2012} or multiparticle atomic
interferometers~\cite{Zoller2012}, and are close to experimental realizability~\cite{Greiner2014}.

\section{Conclusions and outlook}

\label{Sec:conclusion}

In summary, we have shown that a measure of frustration based on the global-to-local incompatibility overlap (global-to-local infidelity)
quantifies efficiently both generic (geometric, common to and equal for classical and quantum systems alike) and genuinely quantum sources of
frustration.

For a large class of Hamiltonians we have proven that the presence of genuine quantum frustration is possible only if the ground states of the
local interaction terms are non degenerate, that is, the rank of the projection operator onto the ground space of the local interacting subsystems
is $r=1$.

In such cases the total frustration can be related, via different upper and lower bounds, with two types of entanglement: the internal (local)
entanglement between the constituents of the local interaction terms, and the block (global) entanglement between the local interaction terms
and the remainder of the total many-body system. Moreover, we have established a ``monogamy of frustration'' relation quantifying the trade-off
on frustration due to the interplay between local and global entanglement.

We have then applied these general results to spin Hamiltonians with competing anti-ferromagnetic interactions on different length scales, thus
featuring simultaneously generic and genuine quantum contributions to the total frustration. We have then shown that the quantum phase
transition to exact ground-state dimerization and valence bond solids corresponds to an exact transition to genuine quantum frustration.

We have investigated the behavior of the total frustration as a function of the physical parameters of a generic class of many-body Hamiltonians.
We have identified three different regimes, two generic ones in which frustration due to geometry dominates and only one of the competing
interaction terms is strongly frustrated while the remaining one is essentially frustration-free, giving rise to phases with a well defined,
classical-like, anti-ferromagnetic order, and an intermediate one in which genuine quantum frustration is non-negligible, the total frustrations
of the competing interaction terms are comparable, and a quantum valence bond solid phase is established with no definite magnetic order.

We have further shown how these different many-body regimes can be clearly identified in an experiment by investigating the behavior of the static
structure factor, a collective observable amenable to direct observation in time-of-flight experiments for quantum spin chains realized by
atom-optical quantum simulators. Finally, we have also briefly discussed the possibility of the direct observation of frustration, as quantified by
our global-to-local incompatibility, by measuring the visibility in experiments with atomic interferometers.

The investigation was carried out for one-dimensional systems with competing nearest neighbor (NN) and next-to-nearest neighbor (NNN) interactions,
but the methodological framework is general and in principle can be applied also to the study of systems with long-range interactions and/or to
higher dimensions and different lattice geometries.

The concepts introduced in the present work might be particularly relevant to the investigation of systems that feature maximally quantum
(nonlocal) and global orders, such as quantum spin liquids~\cite{spinliquids} featuring topological order~\cite{toporder}. Topologically ordered
phases do not break any symmetry of the many-body Hamiltonian, do not admit local order parameters, and feature a robust ground-state topological
degeneracy and a complex pattern of long-range ground-state entanglement. A typical global signature of topological order is a nonvanishing
sub-leading constant contribution to the area-law leading scaling of the ground-state block entanglement, as measured by the von Neumann entropy of
the block reduced density matrix, known as {\em topological entropy}.

On intuitive grounds, the universal measure of frustration, Eq.~(\ref{def_frustratio}), should be perfectly suited for the characterization of
topological order and other maximally quantum orders, as it measures precisely the degree of their incompatibility with the local order induced
by the local interaction terms.

The global-to-local incompatibility overlap is defined in any spatial dimension (at variance with topological entropy, which is not defined for
one-dimensional systems), it is in principle computable in a clearly and unambiguously defined way (topological entropy has to be computed by
complex and extremely delicate subtraction schemes), and is amenable to direct observation by strategies relying on time-of-flight or
interferometric experiments (Von Neumann topological entropy is not directly accessible to experimental observation, and one should rely instead on
the measurement of the two-R\'enyi topological entropy). On this basis, we are currently investigating the scaling behavior of frustration as a
characterization of topological order~\cite{topfrust}.

In general terms, quantum ground states can be classified according to their distinct patterns of long-range entanglement. This provides the most
basic categorization of quantum phases of matter, more fundamental than Landau's symmetry breaking paradigm. It would then be worth pursuing the
investigation of the relation between frustration, topological order, and patterns of genuine multipartite entanglement, along the lines suggested
by some recent studies~\cite{Giampaolo2013,Giampaolo2014,Jindal2014}. Indeed, the observed analogy between two-site dimerization and single-site factorization suggests that $k$-site ground-state factorization might be a useful tool to understand hierarchies of quantum many-body orders in terms of complex patterns of multipartite ground-state entanglement.

{\bf Acknowledgments:} - SMG and FI are grateful to Marcello Dalmonte for valuable discussions. SMG and FI acknowledge the FP7 Cooperation
STREP Project EQuaM - Emulators of Quantum Frustrated Magnetism (GA n. 323714). SMG and BCH acknowledge the Austrian Science Fund (FWF-P23627-N16).

\end{document}